\pdfoutput=1
\newif\ifFull
\Fulltrue
\ifFull
\documentclass[11pt]{article}
\advance \topmargin by -\headheight
\advance \topmargin by -\headsep
\evensidemargin \oddsidemargin
\else
\documentclass{llncs}
\fi
\usepackage{color}
\usepackage{graphicx}
\usepackage{url}
\usepackage[noend]{algorithmic}
\usepackage{amsmath}
\usepackage{times}

\makeatletter
\def\@begintheorem#1#2{\sl \trivlist \item[\hskip \labelsep{\bf #1\ #2:}]}
\def\@opargbegintheorem#1#2#3{\sl \trivlist
      \item[\hskip \labelsep{\bf #1\ #2\ #3:}]}
\makeatother

\newcommand{\sort}{\mbox{o-sort}}

\ifFull
\newtheorem{theorem}{Theorem}
\newtheorem{lemma}[theorem]{Lemma}

\newenvironment{proof}{\noindent{\bf Proof:}}{\hspace*{\fill}\rule{6pt}{6pt}\bigskip}
\fi

\renewcommand{\Pr}{\mathop{\bf Pr}\nolimits}
\newcommand{\E}{\mathop{\bf E}\nolimits}
\ifFull\else
\let\oldendproof\endproof
\def\endproof{\qed\oldendproof}
\renewcommand{\subsection}[1]{\paragraph{#1.}}
\fi

\begin{document}

\title{Privacy-Preserving Access of Outsourced Data \\ 
       via Oblivious RAM Simulation}

\ifFull
\author{
{Michael T. Goodrich} \\
University of California, Irvine 
\and
Michael Mitzenmacher \\
Harvard University}
\else
\author{
{Michael T. Goodrich}$^\dag$ \and
Michael Mitzenmacher$^\sharp$}
\institute{$^\dag$University of California, Irvine \hspace{2em}
	   $^\sharp$Harvard University}
\fi

\date{}

\maketitle 

\begin{abstract}
\ifFull
Suppose a client, Alice, has
outsourced her data to an external storage provider, Bob,
because he has capacity
for her massive data set, of size $n$, whereas her private storage is
much smaller---say,
of size $O(n^{1/r})$, for some constant $r>1$.
Alice trusts Bob to maintain her data, but she would
like to keep its contents private.
She can encrypt her data, of course, but she also wishes to keep her
access patterns hidden from Bob as well.
\fi
We describe schemes
for the \emph{oblivious RAM simulation} problem
with a small logarithmic or
polylogarithmic amortized
increase in access times, with a very high probability of success,
while keeping the external storage 
to be of size $O(n)$.
\ifFull
To achieve this,
our algorithmic contributions include a parallel MapReduce
cuckoo-hashing algorithm 
and an external-memory data-oblivious sorting algorithm.

\fi
\end{abstract}

\section{Introduction}

Suppose Alice owns a large data set, which she 
outsources to an honest-but-curious server, Bob.
\ifFull
Alice trusts Bob to reliably maintain her data, to update it as requested,
and to accurately answer queries on this data.
But she does not trust Bob to keep her information confidential, either
because he has a commercial interest in her
data or because she fears security leaks by Bob (or his employees)
that could allow third parties 
to learn about Alice's data.  
Such suspicions are well-founded, of course, since both scenarios
are quite possible nowadays with respect to well-known 
commercial data outsourcers and email services.

\fi
For the sake of privacy, Alice can, of course, encrypt the 
cells of the
data she stores with Bob. But encryption is not enough, as Alice can 
inadvertently reveal information
about her data based on how she accesses it.
\ifFull
For example, Chen {\it et al.}~\cite{cwwz-sclwa-10} show that 
sensitive information can be inferred from the access patterns of popular
health and financial web sites even if the contents of these communications
are encrypted.
\fi
Thus, we desire that Alice's access sequence 
(of memory reads and writes)
is \emph{data-oblivious}, that is, the probability distribution for Alice's
access sequence should depend only on the size, $n$, 
of the data set and the number of memory accesses.
Formally, we say a computation is data-oblivious if
$\Pr(S\, |\, {\cal M})$,
the probability that Bob sees an access sequence, $S$,
conditioned on a specific configuration of his memory (Alice's
outsourced memory), $\cal M$, satisfies
$
\Pr(S\, |\, {\cal M}) = \Pr(S\, |\, {\cal M}'),
$
for any memory configuration
${\cal M}' \not= {\cal M}$ such that $|{\cal M}'|=|{\cal M}|$.
In particular, Alice's access sequence
should not depend on the values of any 
set of memory cells in the outsourced memory that Bob maintains for Alice.
To provide for full application generality, we assume outsourced data 
is indexed and we allow Alice to make 
arbitrary indexed accesses to this data for queries and updates.
That is, let us assume this outsourced data model is as general as the
random-access machine (RAM) model.
\ifFull
By accessing her data in a data-oblivious fashion, Alice guarantees that no
inferences about her computations are possible beyond those inferable
from time and space bounds, which themselves can be masked through padding.
\fi

Most computations that Alice would be likely to perform on her
outsourced data are not naturally data-oblivious.
We are therefore interested in this paper in simulation
schemes that would allow Alice to make her access sequence data-oblivious
with low overhead.
For this problem, which is known as 
\emph{oblivious RAM simulation}~\cite{go-spsor-96},
we are primarily interested in the case where Alice has a relatively small
private memory, say, of constant size or size that is $O(n^{1/r})$, 
for a given constant $r > 1$.

\subsection{Our Results}
In this paper, we show how Alice can
perform an oblivious RAM simulation,
with very high probability\footnote{We show that our simulation fails to be
  oblivious with negligible probability; that is,
  the probability that the algorithm
  fails can be shown to be
  $O\left( \frac{1}{n^{\alpha}} \right)$ for any $\alpha>0$.  
  We say an event holds with {\it very high probability}
  if it fails with negligible probability.},
with an amortized time overhead of $O(\log n)$
and with $O(n)$ storage overhead purchased from Bob, while using a
private memory of size $O(n^{1/r})$, for a given constant $r>1$. 
With a constant-sized memory, we show that she can do this simulation with
overhead $O(\log^2 n)$, with a similarly high probability of success.
%
At a high level, our result shows that Alice can leverage the privacy
of her small memory 
to achieve privacy in her much larger outsourced data set of size $n$.
\ifFull
For instance, with $r=4$,
Alice could use a private memory of order one megabyte to
achieve privacy in an outsourced memory with
size on the order of one yottabyte.
\fi
Interestingly, our techniques 
involve the interplay of some
seemingly unrelated new results, which may be of
independent interest, including 
an efficient MapReduce parallel algorithm for cuckoo hashing
and
a novel deterministic data-oblivious external-memory sorting algorithm.
\ifFull
We combine these results to achieve our improved oblivious RAM
simulations by exploiting some of the unique performance characteristics of
cuckoo hashing together with efficient data-oblivious methods for performing
cuckoo hashing.
\fi

\subsection{Previous Related Results}
\ifFull
Pippenger and Fischer~\cite{pf-racm-79} are the first to study the
simulation of one computation in a given model with an oblivious one in 
a related model, in that they show how to simulate
a one-tape Turing machine computation of length $n$ with
an oblivious two-tape Turing machine computation of
length $O(n\log n)$, i.e., with an $O(\log n)$ overhead.

\fi
Goldreich and Ostrovsky~\cite{go-spsor-96} introduce
the oblivious RAM simulation problem and show that it requires an
overhead of $\Omega(\log n)$ under some reasonable assumptions about the
nature of such simulations.
For the case where Alice has only a constant-size private
memory,
they show how Alice can easily achieve an overhead of $O(n^{1/2}\log n)$,
with $O(n)$ storage at Bob's server, and,
with a more complicated scheme, how Alice can achieve an overhead of
$O(\log^3 n)$ with $O(n\log n)$ storage at Bob's server.

Williams and Sion~\cite{DBLP:conf/ndss/WilliamsS08}
study the oblivious RAM simulation problem for the case when the data
owner, Alice, has a
private memory of size $O({n}^{1/2})$,
achieving an expected amortized time overhead of $O(\log^2 n)$ using
$O(n\log n)$ memory at the data provider.
Incidentally,
Williams {\it et al.}~\cite{wsc-bcomp-08} claim an
  alternative method that uses an $O({n}^{1/2})$-sized private memory
  and achieves $O(\log n\log\log n)$ amortized time overhead with
  a linear-sized outsourced storage, but some
  researchers (e.g., see~\cite{pr-orr-20}) 
  have raised concerns with the assumptions and analysis of
  this result.

The results of this paper were posted by the authors in preliminary 
form as~\cite{arxiv1007.1259}.
Independently,
Pinkas and Reinman~\cite{pr-orr-20} published an oblivious RAM simulation
result for the case where Alice maintains a 
constant-size private memory, claiming that Alice can achieve
an expected amortized
overhead of $O(\log^2 n)$ while using $O(n)$ storage space at the
data outsourcer, Bob.
Unfortunately, their construction contains a flaw that allows Bob to
learn Alice's access sequence, with high probability, in some cases,
which our construction avoids.

Ajtai~\cite{a-orwca-10} shows how oblivious RAM simulation can be
done with a polylogarithmic
factor overhead 
without cryptographic assumptions about the existence of
random hash functions, as is done in previous 
papers~\cite{go-spsor-96,pr-orr-20,DBLP:conf/ndss/WilliamsS08,wsc-bcomp-08},
as well as
any paper that derives its security or privacy from the random oracle model
(including this paper).
A similar result is also given by Damg\aa{}rd {\it et al.}~\cite{dmn-psor-10}.
Although these results address interesting theoretical limits 
of what is achievable without random hash functions,
we feel that the assumption about
the existence of random hash functions
is actually not a major obstacle in practice, given the ubiquitous
use of cryptographic hash functions.

\ifFull
In addition to the previous upper-bound results,
Beame and Machmouchi~\cite{bm-mrors-10} show that if the
additional space utilized at the data outsourcer (besides the space
for the data itself) is sufficiently sublinear, then the overhead for
oblivious RAM simulation has a superlogarithmic lower bound.
We note that our logarithmic-overhead result does not contradict
this lower bound, however, since we assume
Alice can afford to pay for $O(n)$ additional space at the data outsourcer on
top of the space for her $n$ data items (but
she would prefer, say, not to pay for $O(n\log n)$ additional space).
\fi

\section{Preliminaries}
\ifFull
Before we give the details for our oblivious RAM simulation result,
we need to review some results regarding some other topics.
\fi

\subsection{A Review of Cuckoo Hashing}
Pagh and Rodler \cite{pr-ch-04}
introduce \emph{cuckoo hashing}, which is a hashing scheme using
two tables, each with $m$ cells,
and two hash functions, $h_1$ and $h_2$, one for each table,
where
we assume $h_1$ and $h_2$ can be modeled as random hash functions for the
sake of analysis.
The tables store $n = (1-\epsilon)m$ keys, where one key can be held in
each cell, for a constant $\epsilon<1$.
\ifFull
Note that the total load factor is less than $1/2$.
\fi
Keys can be
inserted or deleted over time; the requirement is that at most 
$n = (1-\epsilon)m$ distinct keys are stored at any time.  A stored key $x$ 
should be located at either $h_1(x)$ or $h_2(x)$, and,
hence, lookups take constant time.  On insertion of a new key $x$, cell
$h_1(x)$ is examined.  If this cell is empty, $x$ is placed in this
cell and the operation is complete.  Otherwise, $x$ is placed in this
cell and the key $y$ already in the cell is moved to $h_2(y)$.  This
may in turn cause another key to be moved, and so on.  We say that a
failure occurs if, for an appropriate constant $c_0$, after 
$c_0 \log n$ steps this process has not
successfully terminated with all keys located in an appropriate cell.
Suppose we insert an $n$th key into the system.  
Well-known attributes of cuckoo hashing include:
\begin{itemize}
\item The expected time to insert a new key is bounded above by a constant.
\item The probability a new key causes a failure is $\Theta(1/n^2)$.  
\end{itemize}

\ifFull
Since their introduction,
several variations of cuckoo hashing have been developed and studied.
For instance, one can
use asymmetrically sized tables, or use a single table for both hash
functions.  Using more than two choices, or $d$-ary cuckoo hashing,
allows for higher load factors but requires more complex insertion
mechanisms \cite{dw-badtpcs-07}. Similarly, one can store more than
one key in a bucket to obtain higher loads at the expense of more
complex insertion mechanisms \cite{fpss-seht-05}.  Other variations
and uses of cuckoo hashing are described in
\cite{m-soq-09}.  Our work will focus on standard cuckoo hashing, for which
we describe a natural parallelization with provable performance guarantees.
\fi
Kirsch, Mitzenmacher, and
Wieder introduce the idea of utilizing a {\em stash}
\cite{kmw-chs-09}.  A stash can be thought of as an additional memory
where keys that would otherwise cause a failure can be placed.  In
such a setting, a failure occurs only if the stash itself overflows.
For $n$ items inserted in a two-table cuckoo hash table, the total
failure probability can be reduced to $O(1/n^{k+1})$ for any constant
$k$ using a stash that can hold $k$ keys.  
\ifFull
We can use this fact to handle
polynomial-time computations using cuckoo hash tables; any fixed
polynomial number of inserts and deletes can be handled with a
constant-sized stash, at the expense of having to check the stash at
each step.  
\fi
For our results, we require a generalization of
this result to stashes that are larger than constant sized.  
\ifFull
We present
this generalization in Appendix~\ref{app:stash} (and make note of where it is required).

There has not been substantial previous work specifically on 
parallel cuckoo hashing,
although historically there has been substantial work on parallel
hashing schemes for load balancing (see, e.g., \cite{acmr,DM2,klm}).  
Recently, Alcantara {\it et al.}~\cite{asasmoa-rtphg-09} sketch a parallel
construction method for a three-table cuckoo hashing scheme, as well
as giving a hybrid parallel-sequential cuckoo hashing method for GPU
computing.  Although they do not theoretically analyze their
fully-parallel algorithm, it is less efficient in terms of its
overall work bounds; hence, it
is less suitable for oblivious simulation than the two-table parallel
cuckoo hashing method we describe here.
Later in this paper, we provide a parallel algorithm for cuckoo
hashing that runs in $O(\log n)$ time and requires $O(n)$ total work, with 
very high probability,
by applying tail bounds on the performance of general
cuckoo hashing to a novel construction in the MapReduce parallel model.
\fi

\ifFull
\subsection{The MapReduce Parallel Model}
Historically, the PRAM model~\cite{j-ipa-92} and the
bulk-synchronous parallel (BSP) model~\cite{v-bmpc-90}
have provided general and powerful parallel models of computation.
In addition, the BSP model has, more recently, been shown to be 
a starting point for models for GPU programming~\cite{hzg-bbsgp-08},
cloud computing~\cite{rpcnh-mnmmb-09}, and
multi-core architectures~\cite{v-bmmcc-08}.
Nevertheless, 
the power and generality of these models makes them challenging to use as
parallel programming paradigms.

As an alternative, the MapReduce programming paradigm 
has been introduced to provide a simple approach to parallel programming
(e.g., see~\cite{dg-msdpl-08,DBLP:conf/soda/FeldmanMSSS08,ksv-amcfm-10}).
This programming paradigm is gaining widespread interest, as it
is used in Google data centers and has been implemented in the 
open-source Hadoop system~\cite{w-hdg-09} for server clusters.
The MapReduce paradigm has also been recently formalized in 
parallel computing models as 
well~\cite{DBLP:conf/soda/FeldmanMSSS08,ksv-amcfm-10}.
It is based on the specification of a computation in terms of a sequence of
map, shuffle, and reduce steps.
\fi

\ifFull
\subsection{Data-Oblivious Sorting}
Data-oblivious sorting has a long and interesting history, starting with the
earliest methods for constructing sorting networks 
(e.g., see~\cite{k-ss-73}).
Thus, for oblivious RAM sorting, there are several existing algorithms.
For instance,
if one is interested in an asymptotically optimal method, then one
can use the AKS sorting network, which is a deterministic
data-oblivious sorting method that runs in $O(n\log n)$
time~\cite{aks-scps-83}.
Unfortunately,
this method is 
not very practical, so in practice it would probably be more
advantageous to use either a simple deterministic suboptimal method,
like odd-even mergesort~\cite{b-snta-68}, which runs in $O(n\log^2 n)$
time, or a simple randomized method, like 
randomized Shellsort~\cite{g-rsaso-10}, which runs in $O(n\log n)$
time and sorts with very high probability.
Both of these methods are data oblivious.

For our main result, however, we desire an efficient method
for data-oblivious sorting in the external-memory model.
In this model, memory is divided between an internal memory
of size $M$ and an external memory (like a disk), which initially
stores an input of size $N$.
The external memory is divided into blocks of size $B$ and 
we can read or write any block in an atomic action called an I/O.
We say that
an external-memory sorting algorithm is \emph{data-oblivious} if the sequence
of I/Os that it performs is independent of the values of the data it is
processing.
Unfortunately, many of the existing efficient external-memory sorting 
algorithms are not oblivious in this sense
(e.g., see~\cite{av-iocsr-88,DBLP:reference/algo/Vitter08}).
Alternatively, existing oblivious external-memory 
sorting algorithms~\cite{cc-scnoa-06} are not
fully scalable to memories of size $O(n^\epsilon)$, or even smaller,
which is what we need here.
So we describe
an efficient external-memory oblivious sorting algorithm,
which is an
\emph{external-memory $k$-way modular mergesort} that
is an external-memory adaptation of Lee and Batcher's
generalization~\cite{lb-ammsn-95} of 
odd-even mergesort~\cite{b-snta-68}.
It uses a data-oblivious sequence of $O((N/B)\log^2_{M/B} (N/B))$ I/Os.

In the sections that follow, we describe how we achieve each of the results
outlined above.
We begin with our parallel algorithm for cuckoo hashing.
\fi

\section{MapReduce Cuckoo Hashing}
\ifFull
In this section, we
describe our MapReduce algorithm for setting up
a cuckoo hashing scheme.
We begin by reviewing the MapReduce paradigm.
\fi

\subsection{The MapReduce Paradigm}
In the MapReduce paradigm
(e.g., see~\cite{DBLP:conf/soda/FeldmanMSSS08,ksv-amcfm-10}),
a parallel computation is 
defined on a set of values, $\{x_1,x_2,\ldots,x_n\}$, and consists of a
series of \emph{map}, \emph{shuffle}, and \emph{reduce} steps:
\begin{itemize}
\item
A map step applies a \emph{mapping function}, $\mu$, 
to each value, $x_i$, to produce a key-value pair, $(k_i,v_i)$.
To allow for parallel execution,
the function,
$\mu(x_i)\rightarrow (k_i,v_i)$, must depend only on $x_i$.
\item
A shuffle step takes all the key-value pairs produced in the previous map
step, and produces a set of lists, $L_k=(k;v_{i_1},v_{i_2},\ldots)$,
where each such list consists of all the values, $v_{i_j}$, such
that $k_{i_j}=k$ for a key $k$ assigned in the map step.
\item
A reduce step applies a \emph{reduction function}, $\rho$, to each list,
$L_k=(k;v_{i_1},v_{i_2},\ldots)$, 
formed in the shuffle step, to produce a set of
values, $w_{1},w_{2},\ldots\,$.
The reduction function, $\rho$, is allowed to be defined sequentially 
on $L_k$, but should be independent of other lists $L_{k'}$ where $k'\not=k$.
\end{itemize}
Since we are using a MapReduce algorithm as a means to an end, rather than as an
end in itself, we allow
values produced at the end of a reduce step to be of two types:
\emph{final values}, which should be included in the output of
such an algorithm when it completes
and are not included in the set of values given as input to the next
map step,
and \emph{non-final values},
which are to be used as the input to
the next map step. 
Thus, for our purposes, a MapReduce computation continues performing map,
shuffle, and reduce steps until the last reduce step is executed, at which 
point we output all the final values produced over the course
of the algorithm.

\ifFull
In the MRC version of this
model~\cite{ksv-amcfm-10}, the computation of $\rho$ is restricted to use only
$O(n^{1-\epsilon})$ working storage, for a constant $\epsilon>0$.
\fi
In the MUD version of this model~\cite{DBLP:conf/soda/FeldmanMSSS08}, 
which we call the \emph{streaming-MapReduce} model,
the computation of $\rho$ is restricted to be a streaming algorithm
that uses only
$O(\log^c n)$ working storage, for a constant $c\ge 0$.
Given our interest in applications to data-oblivious
computations, we define a version that 
further restricts the computation of $\rho$ to be a 
streaming algorithm that uses only $O(1)$ working storage.
That is, we focus on a streaming-MapReduce model where $c=0$,
which we call the \emph{sparse-streaming-MapReduce} model.
In applying this paradigm to solve some problem, we assume we are initially
given a set of $n$ values as input, for which we then perform $t$ steps
of map-shuffle-reduce, as specified by a sparse-streaming-MapReduce
algorithm, $\cal A$.


\ifFull
The performance complexity of a 
MapReduce algorithm can be measured in several 
ways (e.g., see~\cite{DBLP:conf/soda/FeldmanMSSS08,ksv-amcfm-10}).
For instance, we can count the number of map-shuffle-reduce steps, $t$.
In our case, however, we are more interested in the performance issues
involved in data-oblivious simulation of sparse-streaming-MapReduce
algorithms.
\fi
Let us
define the \emph{message complexity} of a MapReduce to be the total
size of all the inputs and outputs to all the map, shuffle, and reduce steps
in a MapReduce algorithm.
That is, if we let $n_i$ denote the total size of the input and output sets for the 
$i$th phase of map, shuffle, and reduce steps,
then the message complexity of a MapReduce algorithm is $\sum_i n_i$.
\ifFull
This correlates to the notion of ``work'' in a 
traditional parallel algorithm~\cite{j-ipa-92}.
\fi

Suppose we have a function, $f(i,n)$, such that $n_i\le f(i,n)$, for
each phase $i$, over all possible executions of a MapReduce algorithm,
$\cal A$, that begins with an input of size $n$.
In this case, let us say that $f$ is a \emph{ceiling function}
for $\cal A$.
Such a function is useful for bounding the worst-case performance overhead 
for a
MapReduce computation\ifFull, especially if it is implemented in a distributed
system like Hadoop~\cite{w-hdg-09}\fi.

%

\subsection{A MapReduce Algorithm for Cuckoo Hashing}
Let us now describe an efficient algorithm for setting up a
cuckoo hashing scheme for a given set, $X=\{x_1,x_2,\ldots,x_n\}$, of items,
which we assume come from a universe that can be linearly ordered in some arbitrary fashion.  
Let $T_1$ and $T_2$ be the two tables that we 
are to use for cuckoo hashing and let $h_1$ and $h_2$ be two
candidate hash functions that we are intending to use as well.

%
%

\ifFull
Before we give the details of our algorithm, 
let us give the intuition behind it. 
\fi
For each $x_i$ in $X$, recall that $h_1(x_i)$ and $h_2(x_i)$ are the two
possible locations for $x_i$ in $T_1$ and $T_2$.
We can define
a bipartite graph, $G$, commonly called the \emph{cuckoo graph},
with vertex sets 
$U=\{h_1(x_i)\colon\ x_i\in X\}$
and
$W=\{h_2(x_i)\colon\ x_i\in X\}$
and edge set
$E=\{(h_1(x_i),h_2(x_i))\colon\ x_i\in X\}$.
That is, for
each edge $(u,v)$ in $E$, there is an associated value $x_i$ such that 
$(u,v)=(h_1(x_i),h_2(x_i))$, with parallel edges allowed.
Imagine for a moment that an oracle identifies for us each connected
component in $G$ and labels each node $v$ 
in $G$ with the smallest item belonging to an edge 
of $v$'s connected component.
Then we could initiate a breadth-first search from the node $u$ in $U$
such that $h_1(x_i)=u$ and $x_i$ is the smallest item in $u$'s connected
component, to define a BFS tree $T$ rooted at $u$.
For each non-root node $v$ in $T$, we can store the item $x_j$ at
$v$, where $x_j$ defines the edge from $v$ to its parent in $T$.
\ifFull
(Similar ideas have appeared in the different context of history-independent
cuckoo hashing \cite{naor-history}.)
\fi

If a connected component $C$ in $G$ is in fact a tree, then this
breadth-first scheme will accommodate all the items associated with
edges of $C$.  Otherwise, if $C$ contains some non-tree edges with
respect to its BFS tree, then we pick one such edge, $e$. All other
non-tree edges belong to items that are going to have to be stored in
the stash. For the one chosen non-tree edge, $e$, we assign $e$'s item
to one of $e$'s endvertices, $w$, and we perform a ``cuckoo'' action
along the path, $\pi$, from $w$ up to the root of its BFS tree, moving
each item on $\pi$ from its current node to the parent of this node on
$\pi$.  Therefore, we can completely accommodate all items associated
with unicyclic subgraphs or tree subgraphs for their connected
components. All other items are stored in the stash.  For cuckoo
graphs corresponding to hash tables with load less than $1/2$, with
high probability there are no components with two or more non-tree edges,
and the stash further increases the probability that, when such edges
exist, they can be handled.


Unfortunately, we don't have an oracle to initiate the above
algorithm. Instead, we essentially
perform the above algorithm in parallel, starting
from all nodes in $U$, assuming they are the root of a BFS tree. Whenever
we discover a node should belong to a different BFS tree, we simply ignore
all the work we did previously for this node and continue the computation for
the ``winning'' BFS tree (based on the smallest item in that connected
component).
\ifFull
In Figure~\ref{fig:bfs}, we describe an efficient MapReduce algorithm for
performing $n$ simultaneous breadth-first searches such that, any time two
searches ``collide,'' the search that started from a lower-numbered vertex is
the one that succeeds.
\else
Consider an efficient MapReduce algorithm for
performing $n$ simultaneous breadth-first searches such that, any time two
searches ``collide,'' the search that started from a lower-numbered vertex is
the one that succeeds.
\fi
We can easily convert this into an algorithm for cuckoo hashing by adding steps
that process non-tree edges in a BFS search. For the first such edge
we encounter, we
initiate a reverse cuckoo operation, to allocate items in the induced cycle.
For all
other non-tree edges, we allocate their associated items to the stash.

Intuitively,
the BFS initiated from the minimum-numbered vertex, $v$,
in a connected component
propagates out in a wave, bounces at the leaves of this BFS tree, returns
back to $v$ to confirm it as the root, and then propagates back down the BFS
tree to finalize all the members of this BFS tree.
Thus, in time proportional to the depth of this BFS tree
(which, in turn, is at most the size of this BFS tree), we will finalize
all the members
of this BFS tree. And once these vertices are finalized, we no longer need to
process them any more.
Moreover, this same argument applies to the modified BFS that performs the
cuckoo actions.
Therefore, we process each connected component in the cuckoo graph in a
number of iterations
that is, in the worst-case, equal to three times 
the size of each such component (since the waves of the BFS move
down-up-down, passing over each vertex three times).

To bound both the time for the parallel BFS algorithm 
to run and to bound its total work, we require bounds on the component
sizes that arise in the cuckoo graph.  Such bounds naturally appear in 
previous analyses of cuckoo hashing.  In particular, 
the following result is proven in \cite{kmw-chs-09}[Lemma 2.4].

\begin{lemma}
\label{lem:cuckoosize}
Let $v$ be any fixed vertex in the cuckoo graph and let $C_v$ be its
component.  Then there exists a constant $\beta \in (0,1)$ such
that for any integer $k \geq 0$, $$ \Pr(|C_v| \geq k) \leq \beta^k. $$
\end{lemma}

More detailed results concerning the asymptotics of
the distribution of component sizes for cuckoo hash tables can 
be found in, for example \cite{drmota2008precise}, 
although the above result is sufficient to
prove linear message-complexity overhead.  

Lemma~\ref{lem:cuckoosize} immediately implies that the MapReduce 
BFS algorithm 
(and the extension to cuckoo hashing) takes $O(\log n)$
time to complete with high probability.
\ifFull
Note that a direct application of Lemma~\ref{lem:cuckoosize} gives
a with very high probability result if we allow $\omega(\log n)$ time to complete;  however, we do not need this result to hold with
very high probability for our oblivious simulation.
  
In addition, we have the following:
\fi
\begin{lemma}
\label{lem:linear}
The message complexity of the MapReduce BFS algorithm is $O(n)$ with 
very high probability.
\end{lemma}
\begin{proof}
The message complexity is bounded by a constant times
$\sum_v |C_v|$, which in expectation is 
$$\E \left [\sum_v |C_v| \right ] =  
\sum_v \E [ |C_v| ]
\leq 2m \sum_{k \geq 0} \Pr(C_v \geq k)
\leq 2m \sum_{k \geq 0} \beta^k 
= O(m).$$


To prove a very high probability bound, we use a variant of Azuma's
inequality specific to this situation.  If all component sizes were
bounded by say $O(\log^2 n)$, then a change in any single edge in the cuckoo
graph could affect $\sum_v |C_v|$ by only $O(\log^4 n)$, and we could
directly apply Azuma's inequality to the Doob martingale obtained by
exposing the edges of the cuckoo graph one at a time.  Unfortunately,
all component sizes are $O(\log^2 n)$ only with very high probability.
However, standard results yield that one can simply add in the
probability of a ``bad event'' to a suitable tail bound, in this case
the bad event being that some component size is larger than $c_1 \log^2
n$ for some suitable constant $c_1$.  Specifically, we directly utilize Theorem 3.7
from \cite{McDiarmid-98-concentration}, which allows us to conclude that if the probability of a bad event is a superpolynomially small $\delta$, 
then 
$$\Pr \left ( \sum_v |C_v| \geq \sum_v \E [ |C_v|] + \lambda \right ) \leq e^{-(2\lambda^2)/(nc_2 \log^4 n)} + \delta,$$
where $c_2$ is again a suitably chosen constant.
Now choosing $\lambda = n^{2/3}$, for example, suffices.
\end{proof} 

\section{Simulating a MapReduce Algorithm Obliviously}
\ifFull
In this section, we describe efficient oblivious simulations for 
sparse-streaming-MapReduce
computations, which imply efficient oblivious constructions for cuckoo
hashing.

\subsection{A Reduction to Oblivious Sorting}
\fi

Our simulation is based on a reduction to oblivious sorting.

\begin{theorem}
\label{thm:mr-sim}
Suppose $\cal A$ is a sparse-streaming-MapReduce algorithm 
that runs in at most $t$ map-shuffle-reduce steps, and
suppose further that we have a ceiling function, $f$, for $\cal A$.
Then we can simulate $\cal A$ in a data-oblivious fashion in the RAM model 
in time $O(\sum_{i=1}^t \sort(f(i,n)))$, where
$\sort(n)$ is the time needed to sort $n$ items in a data-oblivious
fashion.
\end{theorem}
\begin{proof}
Let us consider how we can simulate the map, shuffle, and reduce steps in
phase $i$ of algorithm $\cal A$ in a data-oblivious way.
We assume inductively that we
store the input values for phase $i$ in an array, $X_i$.
Let us also assume inductively that $X_i$ can store 
values that were created as final values the step $i-1$.
A single scan through the first $f(i,n)$ values of 
$X_i$, applying the map function, $\mu$, as we go, produces all the
key-value pairs for the map step in phase $i$ (where we output a dummy value
for each input value that is final or is itself a dummy value). 
We can store each computed 
value, one by one, in an oblivious fashion using an output array $Y$.
We then obliviously sort $Y$ by keys to bring
together all key-value pairs with the same
key as consecutive cells in $Y$ (with dummy values taken to be larger than all
real keys).
This takes time $O(\sort(f(i,n))$. 
Let us then do a scan of the first $f(i,n)$ cells 
in $Y$ to simulate the reduce step.
As we consider each item $z$ in $Y$, we can 
keep a constant number of state variables as registers in our RAM, 
which collectively maintain the key value, $k$,
we are considering, the internal state of registers needed to compute $\rho$
on $z$, and the output values produced by $\rho$ on $z$.
This size bound is due to the fact that $\cal A$ is a
sparse-streaming-MapReduce algorithm.
Since the total size of this state is constant, the total number of output
values that could possibly be produced by $\rho$ on an input $z$ can be
determined a priori and bounded by a constant, $d$.
So, for each value $z$ in $Y$, we write $d$ values to an output array $Z$,
according to the function $\rho$, padding with dummy values if needed.
The total size of $Z$ is therefore $O(d\, f(i,n))$, which is $O(f(i,n))$.
Still, we cannot simply make $Z$ the input for the next map-shuffle-reduce
step at this point, since we need the input array to have at most $f(i,n)$
values.
Otherwise,
we would have an input array that is a factor of $d$ too large for
the next phase of the algorithm $\cal A$.
So we perform a data-oblivious sorting of $Z$, with dummy values taken
to be larger than all real values, and then we copy the first $f(i,n)$
values of $Z$ to $X_{i+1}$ to serve as the input 
array for the next step to continue the
inductive argument.
The total time needed to perform step $i$ is $O(\sort(f(i,n))$.
When we have completed processing of step $t$, we concatenate all the $X_i$'s
together, flagging all the final values as being the output values for the
algorithm $\cal A$, which can be done in a single data-oblivious scan.
Therefore, we can simulate each step of $\cal A$ in a data-oblivious fashion
and produce the output from $\cal A$, as well, at the end.
Since we do two sorts on arrays of size $O(f(i,n))$ in each 
map-shuffle-reduce step, $i$, of $\cal A$, 
this simulation takes time
\ifFull
\[
O(\sum_{i=1}^t \sort(f(i,n))).
\]
\else
$O(\sum_{i=1}^t \sort(f(i,n)))$.
\fi
\end{proof}

\ifFull
In Appendix \ref{app:sort}, we show that by
combining this result with Lemma~\ref{lem:linear} we get the
following:
\else
We can show that by
combining this result with Lemma~\ref{lem:linear} we get the
following:
\fi

\begin{theorem}
\label{thm:mapreduce}
Given a set of $n$ distinct items and corresponding hash values,
there is a data-oblivious algorithm for constructing
a two-table cuckoo hashing scheme of size $O(n)$ with a stash 
of size $s$, whenever this stash size is sufficient,
in $O(\sort(n+s))$ time.
\end{theorem}

\ifFull
Recall that for the constant-memory case, $\sort(n)$ is $O(n\log n)$.
\fi

\subsection{External-Memory Data-Oblivious Sorting}
In this section, we give our efficient external-memory oblivious sorting
algorithm.
Recall that in this model memory is divided between an internal memory
of size $M$ and an external memory (like a disk), which initially
stores an input of size $N$, and that
the external memory is divided into blocks of size $B$, for which
we can read or write any block in an atomic action called an I/O.
In this context,
we say that
an external-memory sorting algorithm is \emph{data-oblivious} if the sequence
of I/Os that it performs is independent of the values of the data it is
processing.
So suppose we are given an unsorted array $A$ of $N$ comparable items
stored in external memory.
If $N\le M$, then we copy $A$ into our internal memory, 
sort it, and copy it back to disk.
Otherwise, we divide $A$ into $k=\lceil (M/B)^{1/3}\rceil$
subarrays of size $N/k$ and recursively sort each subarray.
Thus, the remaining task is to merge 
these subarrays into a single sorted array.

Let us therefore focus on the task of merging $k$
sorted arrays of size $n=N/k$ each.
If $nk\le M$, then we copy all the lists into internal memory, merge
them, and copy them back to disk.
Otherwise,
let $A[i,j]$ denote the $j$th element in the $i$th array. 
We form a set of $m$ new subproblems, where the $p$th subproblem involves
merging the $k$ sorted subarrays defined by $A[i,j]$
elements such that $j\bmod m = p$, for $m=\lceil (M/B)^{1/3}\rceil$.
We form these subproblems by processing each input subarray and filling in the 
portions of the output subarrays from the input, sending full blocks to disk
when they fill up (which will happen at deterministic moments independent of
data values), all of which uses $O(N/B)$ I/Os.
Then we recursively solve all the subproblems.
Let $D[i,j]$ denote the $j$th element in the output of the 
$i$th subproblem.
That is, we can view $D$ as a two-dimensional array, with each row
corresponding to the solution to a recursive merge.

\begin{lemma}
\label{lem:batcher}
Each row and column of $D$ is in sorted order 
and all the elements in column $j$ are
less than or equal to every element in column $j+k$.
\end{lemma}
\begin{proof}
The lemma follows from Theorem~1 of
Lee and Batcher~\cite{lb-ammsn-95}.
\end{proof}

To complete
the $k$-way merge, then, we imagine that we slide an $m\times k$ rectangle
across $D$, from left to right.
We begin by reading into internal memory the first $2k$ columns of $D$.
Next, we sort this set of elements in internal memory and we output the
ordered list of the $km$ smallest elements (holding back a small buffer of
elements if we don't fill up the last block).
Then we read in the next $k$ columns of $D$ (possibly reading in $k$ 
additional blocks for columns beyond this, depending on the block boundaries), 
and repeat the sorting of the items in internal 
memory and outputting the smallest $km$ elements in order.
At any point in this algorithm, we may need to have up to
$2km+(m+2)B$ elements in internal memory, which, under a reasonable tall
cache assumption (say $M>3B^4$), will indeed fit in internal memory.
We continue in this way until we process all the elements in $D$.
Note that, since we process the items in $D$ from left to right
in a block fashion, for all possible data values, the algorithm is
data-oblivious with respect to I/Os.

Consider the correctness of this method.
Let $D_1,D_2,\ldots,D_l$ denote the subarrays of $D$ of size $m\times k$
used in our algorithm.
By a slight abuse of notation, 
we have that $D_1\le D_{3}$,
by Lemma~\ref{lem:batcher}.
Thus, the smallest $mk$ items in $D_1\cup D_2$ are less than or equal 
to the items in $D_{3}$.
Likewise, these $mk$ items are obviously less than the largest $mk$ items in
$D_1\cup D_2$.
Therefore, the first $mk$ items output by our algorithm are the smallest
$mk$ items in $D$.
Inductively, then, we can now ignore these smallest $mk$ items and
repeat this argument with the remaining items in $D$.
\ifFull\else 
Thus, we have the following.
\fi

\ifFull
Let us next consider the I/O complexity of this algorithm.
The $k$-way merge process is a recursive procedure of depth 
$O(\log_{M/B} (N/B))$, such that each level of the recursion is processed
using $O(N/B)$ (data-oblivious) I/Os.
In addition, there are
$O(\log_{M/B} (N/B))$ levels of recursion in the sorting algorithm
itself, for which this $k$-way process performs the merging step in the
divide-and-conquer. Thus, we have the following.
\fi

\begin{theorem}
\label{thm:io}
Given an array $A$ of size $N$
comparable items, we can sort $A$
with a data-oblivious external-memory algorithm that uses
$O((N/B)\log^2_{M/B} (N/B))$ I/Os, under a tall-cache assumption
($M>3B^4$).
\end{theorem}

\ifFull
Thus, we can use this algorithm 
and achieve the stated number of I/Os as an instance of the
$\sort(N)$ function.
We also have the following.
\fi

\begin{theorem}
\label{thm:mapreduce2}
Given a set of $n$ distinct items and corresponding hash values,
there is a data-oblivious algorithm for constructing
a two-table cuckoo hashing scheme of size $O(n)$ with a stash
of size $s = O(\log n)$ whenever this stash size is sufficient,
using a private memory of size $O(n^{1/r})$, for a given fixed
constant $r > 1$, in $O(n+s)$ time.
\end{theorem}
\begin{proof}
Combine Theorems~\ref{thm:mapreduce}
and~\ref{thm:io},
with $N=n+s$, $B=1$, and $M\in O(n^{1/r})$.
\end{proof}

\section{Oblivious RAM Simulations}


Our data-oblivious simulation of a non-oblivious algorithm on a RAM
follows the general approach of
Goldreich and Ostrovsky~\cite{go-spsor-96}, but differs from it in
some important ways, most particularly in our use of cuckoo hashing.
We assume throughout that Alice encrypts the data she
outsources to Bob using
a probabilistic encryption scheme, 
so that multiple encryptions of the same value are extremely
likely to be different. Thus, each
time she stores an encrypted value, there is no way 
for Bob to correlate this value to other values or previous values.
So the only remaining information that needs to be
protected is the sequence of accesses that Alice makes to her data.


Our description simultaneously covers two separate cases:
the constant-sized private memory case with very high probability 
amortized time bounds,
and the case of private memory of size $O(n^{1/r})$ for some constant $r > 1$.
The essential description is the same for these settings, 
with slight differences
in how the hash tables are structured as described below. 

We store the $n$ data items in a hierarchy
of hash tables, $H_k$, $H_{k+1}$, $\ldots$, $H_L$, where 
$k$ is an initial starting point for our hierarchy and
$L=\log n$.
Each table, $H_i$, has capacity for $2^{i}$ items, which are distributed
between ``real'' items, which correspond to memory cells of the RAM, plus 
``dummy'' items, which are added for the sake of obliviousness to
make the number of items stored in $H_i$ be the appropriate value. 
The starting table, $H_k$, is simply an array that we access
exhaustively with each RAM memory read or write.
The lower-level tables, $H_{k+1}$ to $H_l$,
for $l$ determined in the analysis, are standard hash tables with $H_i$
having $2^{i+1}$ buckets of size $O(\log n)$, whereas higher-level tables, 
$H_{l+1}$ to $H_L$, are cuckoo hash tables, with $H_i$ having $(1+\epsilon)2^{i+2}$
cells and a stash of size $s$, where $s$ is determined in the analysis
and $\epsilon>0$ is a constant.
The sizes of the hash tables in this hierarchy increase geometrically; hence
the total size of all the hash tables is proportional to the size of $H_L$,
which is $O(n)$.
Our two settings will differ in the starting points for the various types
of hash tables in the hierarchy as well as the size of the stash associated
with the hash tables.


For each $H_i$ with $i<L$ we keep a count, $d_i$,
of the number of times that $H_i$ has been accessed since 
first being constructed as an 
``empty'' hash table containing $2^i$ dummy values,
numbered consecutively from $-1$ to $-2^{i}$.
For convenience, in what follows, let us think of each hash table
$H_i$ with $i>l$ as being a standard cuckoo hash table,
with a stash of size $s = \Theta(\log n)$ chosen for the sake of a desired
superpolynomially-small error probability.
Initially, every $H_i$ is an empty cuckoo hash table,
except for $H_L$, which contains all $n=2^L$
initial values plus $2^L$ dummy values.

We note that there is, unfortunately, some subtlety in the geometric
construction of hash tables in the setting of constant-sized private
memory with very high probability bounds.  A problem arises in that for small hash tables, of size say
polylogarithmic in $n$, it is not clear that appropriate very high
probability bounds, which required failures to occur with probability
inverse superpolynomial in $n$, hold with logarithmic sized stashes.  Such
results do not follow from previous work
\cite{kmw-chs-09}, which focused on constant-sized stashes.  However,
to keep the simulation time small, we cannot have larger than a
logarithmic-sized stash (if we are searching it exhaustively), and we
require small initial levels to keep the simulation time small.  


\ifFull
In Appendix~\ref{app:stash}, we 
\else
We can
\fi
show that we can cope with the problem
by extending results from \cite{kmw-chs-09} that logarithmic sized
stashes are sufficient to obtain the necessary probability bounds
for hash tables of size that are polylogarithmic in $n$.  In order to
start our hierarchy with small, logarithmic-sized hash tables, we 
simply use standard hash tables as described above for levels $k+1$
to $l=O(\log\log n)$ and use cuckoo hash tables, each with a stash of
size $s=O(\log n)$, for levels $l+1$ to $L$.  



\paragraph{Access Phase.}
When we wish to make an access for a data cell at index $x$, 
for a read or write, we first look in $H_k$ exhaustively to see if it
contains an item with key $x$.
Then, we initialize a flag, ``found,'' to \textbf{false} iff we have
not found $x$ yet.
We continue by performing an access\footnote{This access is actually
  two accesses if the table is a cuckoo hash table.}
to each hash table $H_{k+1}$ to
$H_L$, which is either to $x$ or to a dummy value (if we have already
found $x$).
\ifFull
We give the details of this search in Figure~\ref{fig:access}.
\fi

Our privacy guarantee depends on us 
never repeating a search. That is, we never perform a lookup for the same index,
$x$ or $d$, in the same table, that we have in the past.
Thus, after we have performed the above lookup for $x$,
we add $x$ to the table $H_k$, possibly with a new data value if
the access action we wanted to perform was a write. 

\paragraph{Rebuild Phase.}
With every table $H_i$, we associate a \emph{potential}, $p_i$.
Initially, every table has zero potential.
When we add an element to $H_k$, we increment its potential.
When a table $H_i$ has its potential, $p_i$, reach $2^i$, we 
reset $p_i=0$ and empty the 
unused values in $H_i$ into $H_{i+1}$ and add $2^i$ to $p_{i+1}$.
There are at most $2^i$ such unused values, so we pad this set with
dummy values to make it be of size exactly $2^i$; these dummy values are
included for the sake of obliviousness and can be ignored after
they are added to $H_{i+1}$.
Of course, this 
could cause a cascade of emptyings in some cases, which is fine.

Once we have performed all necessary emptyings,
then for any $j<L$, $\sum_{i=1}^j p_i$ is equal to the
number of accesses made to $H_j$ since it was last emptied.
Thus, we rehash each $H_i$ after it has been accessed $2^i$ times.
Moreover, 
we don't need to explicitly store $p_i$ with its associated hash table,
$H_i$, as $d_i$ can be used to infer the value of $p_i$.

The first time we empty $H_i$ into
an empty $H_{i+1}$,
there must have been exactly $2^i$ accesses made to $H_{i+1}$ since it
was created.
Moreover, the first emptying of $H_i$ into
$H_{i+1}$ involves the addition of $2^i$ values to $H_{i+1}$, some of which
may be dummy values.

When we empty $H_i$ into $H_{i+1}$ for the second time it will actually be
time to empty $H_{i+1}$ into $H_{i+2}$, as $H_{i+1}$ would have been accessed
$2^{i+1}$ times by this point---so we can simply union 
the current (possibly padded) contents of
$H_i$ and $H_{i+1}$ together
to empty both of them into $H_{i+2}$ (possibly with further
cascaded emptyings). 
Since the sizes of hash tables in our hierarchy 
increase geometrically, the size of
our final rehashing problem will always be proportional to the size of 
the final hash table that we want to construct.
Every $n=2^L$ accesses we reconstruct the entire hierarchy, placing 
all the current values into $H_L$.
Thus, the schedule of table emptyings follows a 
data-oblivious pattern depending only on $n$.

\paragraph{Correctness and Analysis.}
To the adversary, Bob, each lookup in a hash table, $H_i$,
is to one random location (if the table is a standard hash table)
or to two random locations 
and the $s$ elements in the stash
(if the table is a cuckoo hash table), which can be 
\begin{enumerate}
\item a search for a real item, $x$, that is not in $H_i$,
\item a search for a real item, $x$, that is in $H_i$,
\item a search for a dummy item, $d_i$.
\end{enumerate}
Moreover, as we search through the levels from $k$ to $L$ we go
through these cases in this order (although for any access, we might
not ever enter case~1 or case~3, depending on when we find $x$).
In addition, if we search for $x$ and don't find it
in $H_i$, we will eventually 
find $x$ in some $H_j$ for $j>i$ and then insert $x$ in $H_k$; hence, 
if ever after this point in time we perform a future
search for $x$, it will be found prior to $H_i$.
In other words,
we will never repeat a search for $x$ in a table $H_i$.
Moreover, we continue performing dummy lookups in tables $H_j$,
for $j>i$, even after we
have found the item for cell $x$ in $H_i$, which are to random locations 
based on a value of $d_i$ that is also not repeated.
Thus, the accesses we make to any
table $H_i$ are to locations that are chosen independently at random.
In addition, so long as our accesses don't violate the possibility of
our tables being used in a valid cuckoo hashing scheme (which our
scheme guarantees with very high probability\ifFull,
  which follows from Appendix~\ref{app:stash}\fi)
then all accesses are
to independent random locations that also happen to correspond to 
locations that are consistent with a valid cuckoo scheme.
Finally, note that we rebuild each hash table $H_i$ after we have made $2^i$
accesses to it. Of course, some of these $2^i$ accesses 
may have successfully found their search key, while others could have been
for dummy values or for unsuccessful searches.
Nevertheless, the collection of $2^i$ distinct keys used
to perform accesses to $H_i$ will either form a standard hash table
or a cuckoo graph that supports a cuckoo hash table,
with a stash of size $s$, w.v.h.p.
Therefore, with very high probability, the adversary will not be able to
determine which among the search keys resulted in values that were found in
$H_i$, which were to keys not found in $H_i$, and which were to dummy values.

Each memory access involves at most $O(s\log n)$ reads and writes, 
to the tables $H_k$ to $H_L$.
In addition, note that each time an item is moved 
into $H_k$, either it or a surrogate dummy value may eventually
be moved from $H_k$ all the way to $H_L$, participating in $O(\log n)$
rehashings, with very high probability.
In the constant-memory case, by Theorem~\ref{thm:mapreduce},
each data-oblivious rehashing of $n_i$ items takes 
$O((n_i+s)\log (n_i+s))$ time.
In addition, in this case,
we use a stash of size $s\in O(\log n)$ and set $l= k+O(\log\log n)$.
In the case of a private memory of size $O(n^{1/r})$, 
each data-oblivious rehashing of $n_i$ items takes $O(n_i)$ time, by
Theorem~\ref{thm:mapreduce2}.
In addition, in this case, we can use a constant-size stash
(i.e., $s=O(1)$), but start with 
$k=(1/r) \log n$, so that $H_k$ fits in private memory (with all the
other $H_i$'s being in the outsourced memory).  The use of 
a constant-sized stash, however, limits us to a result that holds
with high probability, instead of with very high probability.

To achieve very high probability in this latter case, we utilize a
technique suggested in \cite{futurepaper}.  Instead of using a
constant-sized stash in each level, we combine them into a single
logarithmic-sized stash, used for all levels.  This allows the stash
at any single level to possibly be larger than any fixed constant,
giving bounds that hold with very high probability.  Instead of searching
the stash at each level, Alice must load the entire stash into private 
memory and rewrite it on each memory access.  
Therefore, we have the following.
\begin{theorem}
Data-oblivious RAM simulation of a memory of size $n$ can be done
in the constant-size private-memory case with 
an amortized time overhead of $O(\log^2 n)$, with very high
probability.
Such a simulation can be done in the case of a private memory
of size $O(n^{1/r})$ with an amortized time overhead of $O(\log n)$,
with very high probability, for a constant $r > 1$.
The space needed at the server in all cases is $O(n)$.
\end{theorem}

\ifFull
\section{Conclusion and Future Work}
We have given efficient schemes for performing data-oblivious RAM
simulations, which provide significant improvements for both the
constant-memory and sublinear-memory cases. 
Our techniques involved the interplay of a number of new seemingly-unrelated
results for parallel and external-memory algorithms.

Finally, we return to the issue of using completely random hash
functions in our analysis.  The existence of suitable pseudorandom
hash functions (or a random oracle) is a standard assumption in this
area, although very recent work by Ajtai~\cite{a-orwca-10} 
and Damg\aa{}rd {\it et al.}~\cite{dmn-psor-10} 
provide methods for oblivious RAM simulations without a random oracle,
at the cost of a polylogarithmic increase per operation.  We 
expect that similar results would apply here, although the polylogarithmic
increase overhead required for hashing would clearly be significant if
such computations cannot be done completely in private memory
and/or they result in polylogarithmic server accesses for each hash
computation, as in~\cite{a-orwca-10,dmn-psor-10}.
One promising approach, following the work of Arbitman, Naor, 
and Segev~\cite{arbitman2009amortized}, would be to use the result of 
Braverman~\cite{braverman2009poly} to
show that only hash functions that are polylogarithmically independent
in $n$ are required to obtain the appropriate cuckoo hashing
performance.  Such functions can be stored in polylogarithmic space
and evaluated obliviously in polylogarithmic time by reading all of
the coefficients.  As another possible approach, Aum\"{u}ller shows
that explicit hash functions, requiring sublinear space, can be used
to obtain asymptotically the same bounds on the performance of cuckoo
hashing with a stash~\cite{Au}.  However, it is not yet clear these
hash functions can be utilized in an oblivious fashion, and may require
significant private memory to store the hash functions.  As this is not
the focus of our paper, we leave further details as future work.  

There are a number of 
interesting open problems, including the following:
\begin{itemize}
\item
Can one perform data-oblivious RAM simulation in the constant-memory case
with an amortized time overhead of $O(\log n)$?
\item
Is there a
deterministic external-memory data-oblivious sorting method using
$O((N/B)\log_{M/B}(N/B))$ I/Os?
\end{itemize}
\fi

\paragraph{Acknowledgments.}
We thank Radu Sion and Peter Williams for 
helpful discussions about oblivious RAM simulation.
This research was supported in part by the National Science
Foundation under grants 0724806, 0713046, 0847968, 0915922, 0953071, and
0964473.
A full version of this paper is available as~\cite{arxiv1007.1259v2}.

{\raggedright
\bibliographystyle{abbrv}
\bibliography{cuckoo,cuckoo2,crypto,goodrich}
}

\ifFull 

\begin{figure}[p!]
\begin{algorithmic}[100]
\FOR {\textbf{each} vertex $v$}
\STATE Set $v.\mbox{component}=v$
\STATE Set $v.\mbox{visitLabel}=v$
\STATE Set $v.\mbox{parent}=v$
\STATE Set $v.\mbox{active}=\mbox{\bf true}$
\STATE Set $v.\mbox{pending}=\mbox{\bf false}$
\STATE Set $v.\mbox{complete}=\mbox{\bf false}$
\STATE Set $v.\mbox{finished}=\mbox{\bf false}$
\ENDFOR
\FOR {$3D$ iterations, where $D$ is an upper bound on $G$'s diameter}
\FOR {{\bf each} vertex $v$}
\IF {$v.\mbox{active}=\mbox{\bf true}$}
  \STATE Set $v.\mbox{active}=\mbox{\bf false}$.
  \STATE Set $v.\mbox{pending}=\mbox{\bf true}$.
  \FOR {\textbf{each} edge $(v,w)$}
    \IF[$w$'s BFS is no longer valid] {$w.\mbox{visitLabel}>v.\mbox{component}$}
      \STATE Let $w.\mbox{active} = \mbox{\bf true}$.
      \STATE Let $d$ be the minimum component number of a pending 
	      neighbor of $w$.
      \STATE Let $w.\mbox{component} = d$.
      \IF[$v$ ``wants'' to be the parent of $w$] { $v.\mbox{component}=d$}
        \IF { $v$ is minimum numbered vertex s.t. $v.\mbox{component}=d$}
          \STATE Mark $(v,w)$ as a BFS tree edge.
          \STATE Let $w.\mbox{parent}=v$.
        \ELSE
          \STATE Mark $(v,w)$ as a BFS non-tree edge.
	\ENDIF
      \ENDIF
    \ENDIF
  \ENDFOR
  \ENDIF
  \IF {$v.\mbox{pending}=\mbox{\bf true}$ and all of $v$'s children are
	  \textbf{complete}}
     \STATE Set $v.\mbox{complete}=\mbox{\bf true}$.
   \ENDIF
   \IF {$v.\mbox{complete}=\mbox{\bf true}$ and the parent of $v$ is
	   \textbf{finished}}
     \STATE Set $v.\mbox{finished}=\mbox{\bf true}$.
   \ENDIF
   \IF {$v.\mbox{finished}=\mbox{\bf true}$ and the children of $v$ are
	   \textbf{finished}}
     \STATE Mark $v$ as final; the BFS tree is fixed w.r.t.~$v$.
   \ENDIF
\ENDFOR
\ENDFOR
\end{algorithmic}

\caption{\label{fig:bfs} A MapReduce Algorithm for multiple overlapping
	breadth-first searches. The 
	``active'' flag is used to mark vertices that
	are actively pushing a wave of the BFS to the next level. The 
	``pending'' flag is used 
	to identify vertices waiting for their children to
	confirm their all belonging to the same BFS tree. The ``complete'' flag
	confirms this fact, and the ``finished'' flag confirms that the root is
	complete and the BFS tree is done. Note that every operation
	in this algorithm
	involves testing of simple fields, comparisons 
	of local fields with a fixed value (which is the same for all 
	members of the same list in a reduce step),
	or the computation of minimums or sums.
	In addition, we assume that the marking of an edge as a BFS
	edge is an atomic action, so that parallel copies of that same edge
	would not be marked.
	Thus, each step of this algorithm can be implemented in
	the sparse-streaming-MapReduce model.
	}
\end{figure}

\begin{figure}[p!]
\begin{algorithmic}[100]
\FOR {$i=k+1$ to $l$}
\STATE Read $d_i$ and store its value in a private variable, $d$.
\STATE Increment $d_i$'s value in external storage.
\IF {$\mbox{found} = \mbox{\bf false}$}
\STATE Read the entire bucket $H_i[h(x)]$.
\IF {this bucket is holding $x$}
   \STATE Set $\mbox{found} = \mbox{\bf true}$
   \STATE Store $x$ and the contents of data cell $x$ 
          in private storage---we've found it.
   \STATE Revisit $H_i[h(x)]$ marking $x$ as ``used'' 
          and the others as they were.
\ELSE
   \STATE Revisit $H_i[h(x)]$ marking all items as they were.
\ENDIF
\ELSE
  \STATE Read $H_i[h(-d)]$.
  \STATE Revisit $H_i[h(-d)]$, remarking the items in this bucket as they were.
\ENDIF
\ENDFOR
\FOR {$i=l+1$ to $L$}
\STATE Read $d_i$ and store its value in a private variable, $d$.
\STATE Increment $d_i$'s value in external storage.
\IF {$\mbox{found} = \mbox{\bf false}$}
\STATE Read $H_i[h_1(x)]$ and $H_i[h_2(x)]$ and the stash for $H_i$.
\IF {one of these locations, $y$, is holding $x$}
   \STATE Set $\mbox{found} = \mbox{\bf true}$
   \STATE Store $x$ and the contents of data cell $x$ 
          in private storage---we've found it.
   \STATE Revisit $H_i[h_1(x)]$ and $H_i[h_2(x)]$ and the stash for $H_i$,
   marking $y$ as ``used'' and the others as they were.
\ENDIF
\ELSE
  \STATE Read $H_i[h_1(-d)]$ and $H_i[h_2(-d)]$ and the stash for $H_i$.
  \STATE Revisit $H_i[h_1(-d)]$ and $H_i[h_2(-d)]$ and the stash for $H_i$,
  remarking them as they were.
\ENDIF
\ENDFOR
\end{algorithmic}
\caption{\label{fig:access} The oblivious method for accessing
our hierarchy of outsourced hash tables.}
\end{figure}

\clearpage
\begin{appendix}
\section{Proof of Cuckoo Component Size Distribution}
\label{app:cuckoo}
In this appendix, we provide a proof of Lemma~\ref{lem:cuckoosize},
which states that there exists a constant $\beta \in (0,1)$ such
that for any integer $k \geq 0$, $$ \Pr(|C_v| \geq k) \leq \beta^k, $$
where $v$ is any fixed vertex in the cuckoo graph and $C_v$ is its
component.  

\begin{proof}
We use the method of
\cite{devroye2003cuckoo}.  Consider a breadth first search starting
from some vertex $v$.  The number of neighbors of this vertex in the
cuckoo graph is bounded by a binomial random variable $\mbox{Bin}(n,1/m)$.
Now, assuming a neighbor exists, let $X_2$ be the number of additional
edges adjacent to this neighbor; it is again stochastically bounded by
a binomial random variable $\mbox{Bin}(n,1/m)$.  Continuing a breadth first
search in this manner, if we let $X_1, X_2,\ldots$ be random variables
such that $X_i$ represents the number of new adjacent edges explored at the $i$th
vertex of the breadth first search, and let $Y_1, Y_2,\ldots$ be independent binomial
random variables $\mbox{Bin}(n,1/m)$, we have
$$\Pr(|C_v| \geq k) \leq   \Pr \left (\sum_{i=1}^k X_i \geq k \right )
\leq \Pr \left (\sum_{i=1}^k Y_i \geq k \right ) = \Pr \left (\mbox{Bin}(nk,1/m) \geq k \right ).$$
Note $nk/m = k/(1+\epsilon)$.  Assuming $\epsilon \leq 1$ (the result is easily extended
to other cases) we have via a standard Chernoff bound
$$\Pr\left ( \mbox{Bin}(nk,1/m) \geq k \right ) \leq e^{-\epsilon^2 k/(3(1+\epsilon))}.$$  
Letting $\beta =e^{-\epsilon^2 /(3(1+\epsilon))}$ gives the result.
\end{proof}

\section{Proof of Data-Oblivious Cuckoo Hashing Theorem}
\label{app:sort}
In this appendix, we provide a proof of 
Theorem~\ref{thm:mapreduce}, which states that,
given a set of $n$ distinct items,
there is a data-oblivious algorithm for constructing
a cuckoo hash table of size $O(n)$ with a stash of size $s$ in 
$O(\sort(n+s))$ time.

\begin{proof}
Combining Theorem~\ref{thm:mr-sim} and 
Lemma~\ref{lem:linear} we get that there is a
data-oblivious way to simulate the MapReduce cuckoo hash construction
algorithm in $O(\sort(n))$ time.
But the output of this algorithm is not a table of values. It is instead a
set, $S$, of pairs, $(i,x)$, 
where $i$ is an index of a non-empty cell in our hash table $T$,
and $x$ is the value that should be stored in $T[i]$.
We use the convention here that elements belonging to the stash are
identified with pairs of the form $(0,x)$.
(Note: for simplicity, we also 
assume here that the two tables $T_1$ and $T_2$ are
concatenated into one table.)
To convert $S$ to a standard representation for a cuckoo hash table, 
we perform the following
(data-oblivious) algorithm.

\begin{algorithmic}[100]
\FOR {each $(i,x)$ pair in $|S|$}
\STATE Create a tuple, $(i,S,x)$
\ENDFOR
\FOR {$i=1$ to $n$}
  \STATE Create a tuple, $(i,T,0)$
\ENDFOR
\STATE Sort the tuples from $S$ and $T$ (data-obliviously) by first 
and second coordinates, in a list $L$.
\FOR {$i=1$ to $s$}
\STATE Set $H[i]$ to $x$ if $L[i]=(0,S,x)$, otherwise set $H[i]=0$.
\ENDFOR
\STATE Scan $L$ from left-to-right, copying  
the value $x$, from a $(i,S,x)$ tuple, with $i\not=0$, to its
successor tuple, $(i,T,0)$, making it be $(i,T,x)$.
\STATE Sort the set of tuples in $L$ (data-obliviously) by second coordinates.
\STATE Truncate this list of tuples to only contain those with a $T$.
\FOR {each tuple $(i,T,x)$}
  \STATE Set $T[i]=x$
\ENDFOR
\end{algorithmic}
This gives us a standard table format for the cuckoo hashing scheme,
together with its stash, $H$.
\end{proof}

\section{Proof of Bounds on Cuckoo Hashing when the 
Stash is Larger than Constant Size}
\label{app:stash}

We present a sketch of the result demonstrating the value of a stash
that is greater than constant sized, which we require for some of our
results.  We note that the result here essentially follows using
arguments from \cite{kmw-chs-09}, but requires some non-trivial
additional work, which we provide here.  

In what follows let $m$ be the size of each of the two subtables
of the cuckoo hash table.  Recall that
$n$ is the number of keys to be stored by Alice, but the subtables
used in our construction can be much smaller than $n$; indeed, we aim
for $m$ to be polylogarithmic in $n$, while maintaining failure probabilities
that are polynomial in $n$.  This will require stashes of size $O(\log n)$.

To start, we recall some notation from \cite{kmw-chs-09}.  Let $C_v$
denote the (edges of) the connected component of a node $v$.  Then
$\gamma(G)$ be defined to be the smallest number of graph edges which
should be removed from $G$ so that no cycle remains.  (That is,
$\gamma(G)$ is the cyclomatic number of $G$.)  Denote by $T(G)$ the
number of cyclic components in $G$.  A simple fact shown in
\cite{kmw-chs-09} is that the number of items placed in the stash equals $\gamma(G) -
T(G)$.  Bounds in \cite{kmw-chs-09} are obtained by bounding $\gamma(C_v)$,
the excess for a connected component, and applying stochastic dominance arguments.

Our starting point is \cite{kmw-chs-09}[Lemma 2.8], which states
the following.

\begin{lemma}[(Lemma 2.8 of \cite{kmw-chs-09})]
\label{lem: b-bound} For every vertex $v$ and $t,k,n \ge 1$,
\[
   \Pr(\gamma(C_v) \geq t\ |\ |C_v| = k) \leq \left(\frac{3e^5k^3}{m}\right)^t.
\]
\end{lemma}
Recall also we have shown already in Lemma~\ref{lem:cuckoosize} that,
for a constant $\beta$,
$$ \Pr(|C_v| \geq k) \leq \beta^k.$$
Combining this we see that 
\begin{align}
    \Pr(\gamma(C_v) \geq t)
    &\leq \sum_{k=1}^\infty \Pr( \gamma(C_v) \geq t \ |\ |C_v| = k) \cdot \Pr(|C_v| \geq k) \notag \\
    &\leq \sum_{k=1}^\infty \min \left ( \left(\frac{3e^5k^3}{m}\right)^t,1 \right) \cdot \beta^k.
\end{align}
In \cite{kmw-chs-09} it is noted that the right hand side above is $O(m^{-t})$ when
$t$ is a constant, but here, 
where we may have stashes larger than constant size, 
we need to consider $t$ of up to $O(\log n)$.  
However, in the summation above, we need only concern ourselves with values of 
$k$ that are $O(\log^2 n)$, as larger $k$ values
yield superpolynomially small terms in the summation, in which case as long as 
$m$ is $\Omega(\log^7 n)$, then we have 
$\Pr(\gamma(C_v) \geq t)$ is $m^{-\Omega(t)}$.  Specifically, we claim that   
$\Pr(\gamma(C_v) \geq j+1)$ is at most $m^{-1-\alpha j}$ for some constant $\alpha$. 
Note that this derivation requires $m$ to be polylogarithmic in $n$.  

Now, following the derivation at end of Theorem 2.2 of \cite{kmw-chs-09}, we have
that the probability that the stash exceeds a total size $s$, where $s$
should be taken as $\Theta(\log n)$, is given by
\begin{align*}
    \Pr(\gamma(G)-T(G) \ge s)
    &\le \sum_{k=1}^{2m} \binom{2m}{k} s^k m^{-\alpha s-k}\\
    &\le \sum_{k=1}^{2m} \left(\frac{2me}{k}\right)^k s^k m^{-\alpha s-k}\\
    &= m^{-\alpha s} \sum_{k=1}^{2m} \left(\frac{2es}{k}\right)^k \\
    &= m^{-\Omega(s)}.
\end{align*}
The last line follows since the summation is polynomial in $n$, and 
$m^{-\Omega(s)}$ is superpolynomially small in $n$, given our restrictions
on $m$ and $s$.  

(We note there are some changes from the end of Theorem 2.2 of
\cite{kmw-chs-09}; specifically, there is a typographical error, where
in the original theorem they have $k^s$ in place of $s^k$.  The term
is a bound for the number of ways $k$ positive numbers can sum to $s$,
which is easily bounded by $s^k$; the error does not affect the
results of \cite{kmw-chs-09} but the difference is important here.)

It would be interesting to obtain better bounds on stashes of greater
than constant size.  In this case, it could simplify our construction,
if we could instead use cuckoo hash tables from the initial point of 
our construction with constant private memory.

We emphasize that this result that keys can be placed in a cuckoo hash
table with a suitably sized stash with very high probability is used
in two places in our argument.  First, it must apply to the data items
actually stored in the hash table; if the items cannot be stored, the
RAM simulation fails.  Second, it must apply to the sequence of
locations examined during the actual search process.  That is, in
running our simulation, we sometime search for real items and
sometimes for dummy items in the hash tables.  To Bob, searching for
dummy items must not be distinguishable from searching for real items.
In particular, we must have it be the case that in a given cuckoo hash
table, the real and dummy items searched for correspond to a valid
placement in the cuckoo hash table.  As our argument above does not
distinguish between real and dummy items, we can simply union over the
two cases (real items, and the mix of real and dummy items) to
maintain our high probability bound.

\section{A Flaw in the Construction of Pinkas and Reinman}
\label{app:flaw}
As mentioned above,
Pinkas and Reinman~\cite{pr-orr-20} published an oblivious RAM simulation
result for the case where Alice maintains a 
constant-size private memory, claiming that Alice can achieve
an expected amortized
overhead of $O(\log^2 n)$ while using $O(n)$ storage space at the
data outsourcer, Bob.
Unfortunately, their construction contains a flaw that allows Bob to
learn Alice's access sequence, with high probability, in some cases,
which our construction avoids.

Like many of the recent oblivious RAM simulation results, their construction
involves the use of a hierarchy of $O(\log n)$ 
hash tables, which in their case are cuckoo hash tables (without
stashes) that 
start out (at their highest level) being of a constant size and
double in size with every level.

Consider now
a stage in the simulation when most of the levels are
full and the client does a lookup for a set of items,
$x_1,\ldots, x_k$, which
are still on the bottom level, having never been accessed
before. In this construction, Alice
will do a cuckoo lookup for
each of the $x_i$'s in all the levels in this case, finding them only at
the bottom. 
But at the higher (smaller) levels, these $x_i$'s are not in
the cuckoo tables. Moreover, these $x_i$'s were not considered (even
as ``dummy elements'') when
constructing these cuckoo tables. So it is quite likely,
especially for the really small cuckoo tables, that the sequence of
lookups done for $x_1,\ldots,x_k$ will turn out to produce a set of lookups
that cannot possibly be for a set of items that are all contained in
the cuckoo table. That is, they are likely to form a connected
component in the cuckoo graph with more than one cycle (with at least
constant probability for the really small tables). Thus, the
adversary will learn in this case that Alice was searching for at
least one item that was not found in that small table. Put another
way, with reasonably high probability (in terms of $k$), this access
sequence would be distinguishable from one, say, where Alice kept
looking up the same item, $x$, over and over again, for this latter
sequence will always succeed in finding dummy elements in the smaller
tables (hence, will always produce a set of valid cuckoo-table
lookups). 

Incidentally, in a private communication,
we have communicated this concern to Pinkas and
Reinman, and they have indicated that they
plan to repair this flaw in the journal version of their
paper.

\end{appendix}

\fi
\end{document}